\documentclass{article}

\usepackage[totalwidth=420pt,totalheight=625pt]{geometry}

\usepackage{latexsym,amsthm,amsmath,amssymb,url, textcomp}

\newcommand{\Card}[1]{|#1|}

\newcommand{\var}{\text{\normalfont var}}
\newcommand{\sat}{\text{\normalfont sat}}

\newcommand{\ol}[1]{\overline{#1}}

\newcommand{\Prob}{{\mathbb P}}
\newcommand{\Exp}{{\mathbb E}}

\let\epsilon=\varepsilon
\let\phi=\varphi

\newtheorem{lemma}{Lemma}
\newtheorem{corollary}{Corollary}

\newtheorem{theorem}{Theorem}

\newcommand{\ep}{\epsilon}
\newcommand{\MST}{\textsc{Max-$r$-Sat${}_{\text{tlb}}$}}
\newcommand{\MTT}{\textsc{Max-$2$-Sat${}_{\text{tlb}}$}}
\newcommand{\MLT}{\textsc{Max-$r$-Lin2${}_{\text{tlb}}$}}

\begin{document}

\date{}
\title{Solving MAX-$r$-SAT above a Tight Lower Bound%
  \footnote{%
    \textbf{Publication Information:} This is the author's self-archived
    copy of a paper that has been published in \emph{Algorithmica}
    61:638–655, DOI 10.1007/s00453-010-9428-7. The final publication is
    available at www.springerlink.com.
    A preliminary version of this paper has appeared in the proceedings
    of ACM-SIAM Symposium on Discrete Algorithms (SODA 2010).  We extend
    the preliminary version by introducing the notion of a
    bikernelization and using it to prove the existence of a polynomial
    kernel for the {\sc Max-$r$-CSP${}_{\text{tlb}}$} problem introduced
    in Section \ref{sec5}.  We also obtain a quadratic kernel for
    $\MST$.}}

\author{Noga Alon${}^1$,
Gregory Gutin${}^2$, Eun Jung Kim${}^2$, Stefan
Szeider${}^3$, and Anders Yeo${}^2$\\[6pt]
\small ${}^1$  Schools of Mathematics and Computer Science,\\[-3pt]
\small  Tel Aviv University, Tel Aviv 69978, Israel\\[-3pt]
\small \texttt{nogaa@post.tau.ac.il}\\
\small ${}^2$  Department of Computer Science,
\small  Royal Holloway, University of London\\[-3pt]
\small Egham, Surrey TW20 0EX, UK\\[-3pt]
\small \texttt{\{gutin|eunjung|anders\}@cs.rhul.ac.uk}\\
\small ${}^3$ Institute of Information Systems,
\small Vienna University of Technology\\[-3pt]
\small A-1040 Vienna, Austria\\[-3pt]
\small \texttt{stefan@szeider.net}
}

\maketitle

\thispagestyle{empty}

\newenvironment{compress}{\baselineskip=10pt}{\par}
\begin{abstract}
\noindent
  We present an exact algorithm that decides, for every fixed $r
  \geq 2$ in time $O(m) + 2^{O(k^2)}$ whether a given multiset of
  $m$ clauses of size~$r$ admits a truth assignment that satisfies
  at least $((2^r-1)m+k)/2^r$ clauses.  Thus \textsc{Max-$r$-Sat} is
  fixed-parameter tractable when parameterized by the number of satisfied
  clauses above the tight lower bound~$(1-2^{-r})m$.  This solves an open
  problem of Mahajan, Raman and Sikdar (J. Comput. System Sci., 75, 2009).

  Our algorithm is based on a polynomial-time data reduction procedure
  that reduces a problem instance to an equivalent algebraically
  represented problem with $O(9^rk^2)$ variables.
 This is done by representing the instance as an appropriate polynomial,
 and by applying a probabilistic argument combined with some simple tools
from Harmonic analysis to show that if the polynomial cannot be reduced
to one of size $O(9^rk^2)$, then there is a truth assignment satisfying
the required number of clauses.

  We introduce a new notion of bikernelization from a parameterized
problem to another one and apply it to
prove that the above-mentioned parameterized \textsc{Max-$r$-Sat}
  admits a polynomial-size kernel.

  Combining another probabilistic argument with tools from graph
  matching theory and signed graphs, we show that if an instance of
  \textsc{Max-$2$-Sat} with $m$ clauses has at least $3k$ variables
  after application of a certain polynomial time reduction rule to it,
  then there is a truth assignment that satisfies at least $(3m+k)/4$
  clauses.

  We also outline how the fixed-parameter tractability and polynomial-size kernel results on
  \textsc{Max-$r$-Sat} can be extended to more general families
of Boolean Constraint
  Satisfaction Problems.
\end{abstract}
\vfill

\section{Introduction}\label{section:intro}

The Maximum $r$-Satisfiability Problem (\textsc{Max-$r$-Sat}) is a
classic optimization problem with a wide range of real-world
applications. The task is to find a truth assignment to a multiset of
clauses, each with exactly $r$ literals, that satisfies as many clauses
as possible, or in the decision version of the problem, to satisfy at
least $t$ clauses where $t$ is given with the input. Even
\textsc{Max-2-Sat} is NP-hard~\cite{GareyJohnsonStockmeyer76} and APX-hard ~\cite{Hastad01}, in strong contrast with \textsc{2-Sat}
which is solvable in linear time~\cite{AspvallPlassTarjan79}.

It is always possible to satisfy a $1-2^{-r}$ fraction of a given multiset of
clauses with exactly $r$ literals each; a truth assignment that meets
this lower bound can be found in polynomial time by Johnson's
algorithm~\cite{Johnson73}. This lower bound is \emph{tight} in the
sense that it is optimal for an infinite sequence of instances. In this
paper we show that for every fixed $r$ we can decide in time $O(m) +
2^{O(k^2)}$ whether a given multiset of $m$ clauses admits a truth assignment
that satisfies at least $((2^r-1)m+k)/2^r$ clauses. Thus,
\textsc{Max-$r$-Sat} is fixed-parameter tractable when parameterized by
the number of satisfied clauses above the tight lower bound; this
answers a question posed by Mahajan, Raman and
Sikdar~\cite{MahajanRamanSikdar09}.

Our algorithm described in Section~\ref{sec3} is based on a
polynomial-time data reduction procedure that reduces a problem instance
to an equivalent algebraically represented problem with $O(k^2)$
variables. This is done by representing the instance as an appropriate
polynomial, and by applying a probabilistic argument combined with some
simple tools from Harmonic analysis to show that if the polynomial
cannot be reduced to one of size $O(k^2)$, then there is a truth
assignment satisfying the required number of clauses.  The basic
approach is based on the ideas  of \cite{AlonGutinKrivelevich04},
and a similar
one which, however, does not apply any algebraic reductions, was used in
\cite{GutinKimMnichYeo,GutinKimSzeiderYeo09a} to show the existence of
kernels of size $O(k^2)$ for other problems parameterized above tight lower
bounds.

We also show that the above-mentioned parameterized \textsc{Max-$r$-Sat}
admits a polynomial-size kernel. This can be deduced
from our fixed-parameter result
and a general lemma proved in Section \ref{sec:bik}, and can also
be proved by a more efficient, direct argument.
The lemma, which is interesting in its own right, links
a new concept that we call bikernelization with the well-known
concept of kernelization.
We believe that bikernelization, in general, and the lemma,
in particular, will have further
applications.

In Section~\ref{sec4}, combining another probabilistic argument with
tools from graph matching theory and signed graphs, we show that if an
instance $\cal I$ of \textsc{Max-$2$-Sat} on $m$ clauses has at least
$3k$ variables after application of a certain polynomial time reduction
rule to it, then there is a truth assignment for $\cal I$ that
satisfies at least $(3m+k)/4$ clauses.  Thus, \textsc{Max-$2$-Sat}
admits a problem kernel with at most $3k-1$ variables.

In Section~\ref{sec5}, we outline
how the fixed-parameter tractability and polynomial-size
kernel results on \textsc{Max-$r$-Sat} can
be extended to more general families of
Boolean Constraint Satisfaction Problems. 

In Section \ref{sec:d}, we have a short discussion of the practicality of our results 
and mention a very recent improvement of our kernel result for 
\textsc{Max-$r$-Sat} parameterized above the tight lower bound.

\paragraph{Related Work} Parameterizations \emph{above a guaranteed value} were first considered
by Mahajan and Raman \cite{MahajanRaman99} for the problems
\textsc{Max-Sat} and \textsc{Max-Cut}. They devised an algorithm for
\textsc{Max-Sat} with running time $O^*(1.618^k+\sum_{i=1}^m|C_i|)$ that
finds, for a multiset $\{C_1,\ldots ,C_m\}$ of $m$ clauses, a truth assignment
satisfying at least $\lceil m/2 \rceil +k$ clauses, or decides that no
such truth assignment exists ($|C_i|$ denotes the number of literals in
$C_i$).  In a recent paper~\cite{MahajanRamanSikdar09}, Mahajan, Raman
and Sikdar argue, in detail, that a practical (and challenging) parameter for a
maximization problem is the number of clauses satisfied above a
\emph{tight} lower bound, which is $(1-2^{-r})m$ for \textsc{Max-Sat} if
each clause contains exactly $r$ different variables. Only a few
non-trivial complexity results are known for problems parameterized above a tight
lower
bound~\cite{GutinKimSzeiderYeo09a,GutinRafieySzeiderYeo07,GutinSzeiderYeo08,HeggernesPaulTelleVillanger07,MahajanRaman99}.

Mahajan et al.~\cite{MahajanRamanSikdar09} state several problems
parameterized above a tight lower bound whose parameterized complexity
is open. One of the problems is the (exact) \textsc{Max-$r$-Sat} problem
(an instance consists of $m$ clauses, each containing exactly $r$
different literals) parameterized by the number of satisfied
clauses above the tight lower bound
$(1-2^{-r})m$, i.e., {\MST}.

\section{Preliminaries}

We assume an infinite supply of propositional \emph{variables}. A
\emph{literal} is a variable $x$ or its negation $\ol{x}$.  A
\emph{clause} is a finite set of literals not containing a complementary
pair $x$ and~$\ol{x}$.  A clause is of \emph{size} $r$ if it contains
exactly $r$ literals. For simplicity of presentation, we will denote a
clause by a sequence of its literals. For example, the clause $\{\ol{x},
y\}$ will be denoted $\ol{x}y$ or equivalently $y\ol{x}$.  A \emph{CNF
  formula} $F$ is a finite multiset of clauses (a clause may appear in
the multiset several times).  A variable $x$ \emph{occurs} in a clause
if the clause contains $x$ or $\ol{x}$, and $x$ occurs in a CNF formula
$F$ if it occurs in some clause of~$F$. Let $\var(C)$ and $\var(F)$
denote the sets of variables occurring in $C$ and $F$, respectively.  A
CNF formula is an $r$-CNF formula if $\Card{C}=r$ for all $C\in F$.
Thus we require that each clause of a $r$-CNF formula contains
\emph{exactly}~$r$ different literals (some authors use for that the
term \emph{exact} $r$-CNF).
A \emph{truth assignment} is a mapping $\tau:\ V\rightarrow \{-1,1\}$
defined on some set $V$ of variables. In order to obtain a 'normalized' algebraic representation, we use $\{-1,1\}$ instead of the usual $\{0,1\}$ binary symbols. We write $2^V$ to denote the set
of all truth assignments on $V$.  A truth
assignment $\tau$ \emph{satisfies} a clause $C$ if there is some
variable $x\in C$ with $\tau(x)=1$ or a negated variable $\ol{x}\in C$
with $\tau(x)=-1$.  We write $\sat(\tau,F)$ for the number of clauses
of~$F$ that are satisfied by $\tau$, and we write
\[
\sat(F)=\max_{\tau\in 2^{\var(F)}}\sat(\tau,F).
\]


A \emph{parameterized problem} is a subset $L\subseteq \Sigma^* \times
\mathbb{N}$ over a finite alphabet $\Sigma$. $L$ is
\emph{fixed-parameter tractable} if the membership of an instance
$(x,k)$ in $\Sigma^* \times \mathbb{N}$ can be decided in time
$|x|^{O(1)} \cdot f(k)$ where $f$ is a computable function of the
parameter~\cite{DowneyFellows99,FlumGrohe06,Niedermeier06}.

If the nonparameterized version of $L$ is NP-hard, then $f(k)$ is superpolynomial provided  P$\neq$NP. Often $f(k)$ is `moderately' exponential, which makes the problem practically solvable for small values of $k$. Thus, it is important to parameterize a problem in such a way that the instances with small values of $k$ are of interest.

Given a parameterized problem $L$,
a \emph{kernelization of $L$} is a polynomial-time
algorithm that maps an instance $(x,k)$ to an instance $(x',k')$ (the
\emph{kernel}) such that (i)~$(x,k)\in L$ if and only if
$(x',k')\in L$, (ii)~ $k'\leq f(k)$, and (iii)~$|x'|\leq g(k)$ for some
functions $f$ and $g$. The function $g(k)$ is called the {\em size} of the kernel.
The notion of a kernelization was introduced by Downey and Fellows \cite{DowneyFellows95}. They showed that
a decidable parameterized problem is fixed-parameter
tractable if and only if it admits a
kernelization. Recently, Bodlaender et al. \cite{BDFH} obtained a
framework to give evidence that fixed-parameter tractable problems
do not have a kernel of polynomial size. For excellent overviews
of much recent work on kernelization, see \cite{Bodlaender09,GN}.

We shall consider the following parameterized version of
\textsc{Max-$r$-Sat}.

\begin{center}
\fbox{~\begin{minipage}{11cm}
  \textsc{Max-$r$-Sat above Tight Lower Bound} (or {\MST} for short)

  \smallskip

  \emph{Instance:} A pair $(F,k)$ where $F$ is a multiset of $m$
  clauses of size $r$ and $k$ is a nonnegative integer.

  \smallskip

  \emph{Parameter:} The integer $k$.

\smallskip

  \emph{Question:} Is $\sat(F)\geq ((2^r-1)m+k)/2^r$?

\smallskip
\end{minipage}~}
\end{center}

\noindent We note that Mahajan et al.~\cite{MahajanRamanSikdar09} use a
slightly different formulation of the problem, asking for an assignment
that satisfies at least $(1-2^{-r})m+k$ clauses; since $r$ is fixed,
this change does not affect the complexity of the problem.

We will also refer to the following special case
of another problem introduced by Mahajan et al.~\cite{MahajanRamanSikdar09}.

\begin{center}
\fbox{~\begin{minipage}{11cm}
  \textsc{Max-$r$-Lin2 above Tight Lower Bound} (or {\MLT} for short)

  \smallskip

  \emph{Instance:} A system of $m$ linear equations
$e_1,\dots, e_m$ in $n$ variables over $\mathbb{F}_2$, where
no equation has more than $r$ variables,
and each equation $e_j$ has a positive integral weight $w_j$,
and a nonnegative integer $k$.

  \smallskip

  \emph{Parameter:} The integer $k$.

\smallskip

  \emph{Question:} Is there an assignment of values to the $n$ variables such that the total weight of the satisfied
equations is at least $(W+k)/2$, where  $W=w_1+\cdots +w_m$ ?

\smallskip
\end{minipage}~}
\end{center}

Note that trivially $W/2$ is indeed a tight lower bound
for the above problem, as the expected number of satisfied
equations in a random assignment is $W/2$, and if the equations
come in identical pairs with contradicting right-hand sides, no
assignment satisfies more equations.
It was proved in \cite{GutinKimSzeiderYeo09a}
that ${\MLT}$ admits a kernel with
$O(k^2)$ equations and variables.

We conclude the section by outlining the very basic principles of the probabilistic method which will be implicitly used in this paper. Given random variables $X_1, \ldots , X_n$, the fundamental property known as {\em linearity of expectation} states that $\Exp(X_1+\ldots +X_n)=\Exp(X_1)+ \ldots + \Exp(X_n)$. The {\em averaging argument} utilizes the fact that there is a point for which $X\geq \Exp(X)$ and a point for which $X\leq \Exp(X)$ in the probability space. Lastly, a positive probability $\Prob(A)>0$ for some event $A$ means that there is at least one point in the probability space which belongs to $A$. For example, $\Prob(X\geq k)>0$ tells us that there exists a point for which $X\geq k$. For further reading on the probabilistic method, we refer the reader to the textbook \cite{AS} by Alon and Spencer.

\section{Bikernelization}\label{sec:bik}

In this section we introduce a new notion called bikernelization and study its basic properties.
A bikernelization from $L$ to $L'$ is of interest especially when $L'$ is a well-studied problem.

Given a pair $L,L'$ of parameterized problems,
a \emph{bikernelization from $L$ to $L'$} is a polynomial-time
algorithm that maps an instance $(x,k)$ to an instance $(x',k')$ (the
\emph{bikernel}) such that (i)~$(x,k)\in L$ if and only if
$(x',k')\in L'$, (ii)~ $k'\leq f(k)$, and (iii)~$|x'|\leq g(k)$ for some
functions $f$ and $g$. The function $g(k)$ is called the {\em size} of the bikernel.
Observe that a {\em kernelization} of a parameterized problem
$L$ is simply a bikernelization from $L$ to itself, i.e., a bikerenelization generalizes a kernelization.

Recall that a decidable parameterized problem is fixed-parameter
tractable if and only if it admits a
kernelization. This result can be extended as follows:
A decidable parameterized problem $L$ is fixed-parameter
tractable if and only if it admits a
bikernelization from itself to a decidable parameterized problem $L'$.
Indeed, if $L$ is fixed-parameter
tractable, then $L$ is decidable and admits a
bikernelization to itself. If $L$ is decidable and admits a
bikernelization from itself to a parameterized problem $L'$,
then $(x,k)$ can be decided by first mapping it to $(x',k')$ in polynomial
time and then deciding $(x',k')$ in time depending only on $k'$, and
thus only on $k$.

We are especially interested in cases when kernels are of polynomial
size. The next lemma is similar to Theorem 3 in \cite{BTY09}.

\begin{lemma}\label{lem:pk}
Let $L,L'$ be a pair of decidable parameterized problems such that $L'$
is in NP, and $L$ is NP-complete.
If there is a bikernelization from $L$ to $L'$ producing a
bikernel of polynomial size, then $L$ has a polynomial-size kernel.
\end{lemma}
\begin{proof}
Consider a bikernelization from $L$ to $L'$ that maps an instance
$(x,k)\in L$ to an instance $(x',k')\in L'$ with $k'\le f(k).$
Since $L'$ is in NP and $L$ is NP-complete,
there exists a polynomial time reduction from
$L'$ to $L$. Thus, we can find in polynomial time
an instance $(x'',k'')$ of $L$ which is decision-equivalent with
$(x',k')$, and in turn with $(x,k)$. Observe
that $|x''|\le |x'|^{O(1)}\le k^{O(1)}$ and
$k''\le (k')^{O(1)}+(|x'|)^{O(1)} \le f(k)^{O(1)}+k^{O(1)}.$
Thus, $(x'',k'')$ is a kernel of $L$ of polynomial size.
\end{proof}

\section{MAX-$r$-SAT}\label{sec3}

In this section we describe a polynomial-time data reduction that reduces an instance of $\MST$ into an equivalent algebraically represented problem. The equivalent algebraically represented problems is `normalized' in a sense, which enables us to obtain a bound on the size of a given instance. Some results from probability theory and Harmonic analysis in boolean functions play a central role in proving such a bound. As a result, we prove that $\MST$ is fixed-parameter tractable and in particular we present a quadratic kernel using the notion of bikernelization introduced in the previous section.

\subsection{An Algebraic Representation}\label{algerep}

Let $F$ be an $r$-CNF formula with clauses $C_1,\ldots ,C_m$
in the variables $x_1, x_2, \ldots ,x_n$.

For $F$, consider
\[
X=\sum_{C\in F}[1-\prod_{x_i\in \var(C)}(1+\epsilon_ix_i)],
\]
where $\epsilon_i\in \{-1,1\}$ and $\epsilon_i=-1$ if and only if $x_i$
is in~$C$.

\begin{lemma}\label{lem0}
  For a truth assignment $\tau$, we have
  $X=2^r(\sat(\tau,F)-(1-2^{-r})m)$.
\end{lemma}
\begin{proof}
  Observe that $\prod_{x_i\in \var(C)}(1+\epsilon_ix_i)$ equals $2^r$ if
  $C$ is falsified and 0, otherwise. Thus, $X=m-2^r(m-\sat(\tau,F))$
  implying the claimed formula.
\end{proof}

After algebraic simplification $X=X(x_1, x_2, \ldots ,x_n)$
can be written as
$\label{Xeq}X=\sum_{I\in {\cal S}}X_I,$ where $X_I=c_I\prod_{i\in
I}x_i$, each $c_I$ is a nonzero integer and $\cal S$ is a family of
nonempty subsets of $\{1,\ldots ,n\}$ each with at most $r$
elements.

The question we address is that of deciding whether or not there are
values $x_i \in \{-1,1\}$ so that $X=X(x_1,x_2, \ldots ,x_n) \geq
k$. The idea is to use a probabilistic argument and show that if the
above polynomial has many nonzero coefficients, that is, if $|{\cal S}|$
is large, this is necessarily the case, whereas if it is small, the
question can be solved by checking all possibilities of the relevant
variables.

\subsection{The Properties of $X$}

In what follows, we assume that each variable $x_i$ takes its values
uniformly at random and independently in $\{-1,1\}$ and thus $X$ is a random
variable. Our approach is similar to the one
in~\cite{AlonGutinKrivelevich04}. For completeness, we reproduce part of
the argument (modifying it a bit and slightly improving the constant for
the case considered here).  We need the following simple lemma.

\begin{lemma}[\normalfont see, e.g., \cite{AlonGutinKrivelevich04}, Lemma 3.1]
  \label{l91}
  For every real random variable $X$ with finite and positive fourth
  moment,
  \[
  \Exp(|X|) \geq \frac{\Exp(X^2)^{3/2}}{\Exp(X^4)^{1/2}}.
  \]
\end{lemma}
The above lemma implies the following (see
\cite{AlonGutinKrivelevich04}, Lemma 3.2, part (ii) for a similar
result).
\begin{corollary}
\label{c93}
Let $X$ be a real random variable and suppose that its first, second and
fourth moments satisfy $\Exp(X)=0$, $\Exp(X^2)=\sigma^2>0$ and $\Exp(X^4)
\leq b \sigma^4$, where $b$ is a positive constant. Then
\[
\Prob(X \geq \frac{\sigma}{2 \sqrt b})>0.
\]
\end{corollary}
\begin{proof}
  By Lemma~\ref{l91}, $\Exp(|X|) \geq \frac{\sigma}{\sqrt b}$. Since
  $\Exp(X)=0$ it follows that
  \begin{equation}
    \label{e91}
    \Prob(X>0)\cdot \Exp(X| X >0) \geq
    \frac{\sigma}{2 \sqrt b}.
\end{equation}
Therefore, $X$ must be at least $\sigma/(2 \sqrt b)$ with positive
probability.
\end{proof}

\noindent We also use the hypercontractive inequality, see ~\cite{ODonnell08}. The following lemma states a special case of it first proved by Bonami \cite{Bo70}.

\begin{lemma}[Hypercontractive Inequality]  \label{l92}
  Let $f=f(x_1,\ldots,x_n)$ be a multilinear polynomial of degree (at most) $r$ in $n$
  variables $x_1,\ldots,x_n$ with domain $\{-1,1\}$.  Define a random
  variable $X$ by choosing a vector $(\ep_1,\ldots,\ep_n)\in \{-1,1\}^n$
  uniformly at random and setting $X=f(\ep_1,\ldots,\ep_n)$. Then,
  $\Exp(X^4 ) \leq 9^r (\Exp(X^2))^{2}$.

\end{lemma}

\noindent Returning to the \hbox{random variable $X\!=\!X(x_1,x_2,
  \ldots ,x_n)$} defined in the previous subsection, we prove the
following.

\begin{lemma}
  \label{l94}
  Let $X= \sum_{I\in {\cal S}}X_I,$ where $X_I=c_I\prod_{i\in I}x_i$ is
  as in the previous subsection, and assume it is not identically
  zero. Then $\Exp(X)=0$, $\Exp(X^2)=\sum_{I\in {\cal S}}c^2_I \geq
  |{\cal S}| >0$ and $\Exp(X^4)\le 9^r\Exp(X^2)^2$.
\end{lemma}
\begin{proof}
  Since the $x_i$'s are mutually independent, $\Exp(X)=0$. Note that for $I,
  J\in {\cal S}$, $I\neq J$, we have
  $\Exp(X_IX_J)=c_Ic_J\Exp(\prod_{i\in I\mathrel{\Delta} J}x_i)=0,$ where
  $I\mathrel{\Delta} J$ is the symmetric difference of $I$ and~$J$. Thus,
  $\Exp(X^2)=\sum_{I\in {\cal S}}c^2_I$. By Lemma~\ref{l92},
  $\Exp(X^4)\le 9^r\Exp(X^2)^2$.
\end{proof}

\subsection{The Main Result for General  $r$}
\begin{theorem}
  \label{t95}
  The problem $\MST$ is fixed-parameter tractable and can be solved in
  time $O(m) + 2^{O(k^2)}$. Moreover, there exist
(i) a polynomial-size bikernel from {\MST} to {\MLT},
  and (ii) a polynomial-size kernel of {\MST}. In fact, there are
such a bikernel and a kernel of size $O(k^2)$.
\end{theorem}
\vspace{-5pt}
\begin{proof}
  By Lemma~\ref{lem0} our problem is equivalent to that of deciding
  whether or not there is a truth assignment to the variables $x_1, x_2,
  \ldots, x_n$, so that
\begin{equation}
\label{e1}
X(x_1, \ldots ,x_n) \geq k.
\end{equation}
Note that in particular this implies that if
  $X$ is the zero polynomial, then any truth assignment satisfies
  exactly a $(1-2^{-r})$ fraction of the original clauses.  By
  Corollary~\ref{c93} and Lemma~\ref{l94}, $\Prob(X \geq
  \frac{\sqrt{\Exp(X^2)}}{2\sqrt{b}})>0,$ where $b=9^r$ and
  $\Exp(X^2)=\sum_{I\in {\cal S}}c^2_I \geq |{\cal S}|$; the last
  inequality follows from the fact that each $|c_I|$ is a positive
  integer. Therefore $\Prob(X \geq \frac{\sqrt{ |{\cal S}|}}{2\cdot
    3^r})>0$. Now, if $k\le \frac{\sqrt{|{\cal S}|}}{2\cdot 3^r}$ then
there are $x_i \in \{-1,1\}$ such that
  (\ref{e1}) holds,
and there is an assignment for which the answer to $\MST$
  is {\sc Yes}. Otherwise, $|{\cal S}| = O(k^2)$, and in fact even
$\sum_{I\in {\cal S}}|c_I| \leq \sum_{I\in {\cal S}}c^2_I =O(k^2)$,
that is, the total number of terms of the simplified polynomial,
even when counted with multiplicities, is at most $O(k^2)$.

  For any fixed $r$, the representation of a problem instance of $m$
  clauses as a polynomial, and the simplification of this polynomial,
  can be performed in time $O(m)$. If the number of nonzero terms of
  this polynomial is larger than $4 \cdot 3^{2r} k^2$, then the answer
  to the problem is {\sc Yes}. Otherwise, the polynomial has at most
  $O(k^2)$ terms and depends on at most $O(k^2)$ variables, and its
  maximum can be found in time $2^{O(k^2)}$.

This completes the proof of the first part of the theorem.
We next establish the second part. Given the simplified  polynomial
$X$ as above, define a problem in {\MLT} with the variables
$z_1, z_2, \ldots ,z_n$ as follows. For each
nonzero  term $c_I \prod_{i \in I} x_i$  consider the linear equation
$\sum_{i \in I} z_i =b$, where $b=0$ if $c_I$ is positive, and $b=1$
if $c_I$ is negative, and either associate this equation with the
weight $w_I=|c_I|$, or duplicate it $|c_I|$ times.  It is easy to check
that this system of equations has an assignment $z_i$ satisfying at least
$[\sum_{I \in {\cal S}} w_I+k]/2$ of the equations if and only if
there are $x_i \in \{-1,1\}$ so that $X(x_1,x_2, \ldots ,x_n)
\geq k$. This is shown by the transformation $x_i=(-1)^{z_i}$. See
also \cite{HastadVenkatesh02} and \cite{GutinKimSzeiderYeo09a} for
a similar discussion. Since, as explained above,  we may assume
that $\sum_{I \in {\cal S}} |c_I| =O(k^2)$ (as otherwise  we know
that the answer to our problem is  {\sc Yes}), this provides the
required bikernel of size $O(k^2)$ to {\MLT}.

It remains to prove the existence of a polynomial size kernel for
the original problem. One way to do that is to apply
Lemma \ref{lem:pk}. Indeed, {\MLT} is in NP, and
{\MST} is NP-complete,
implying the desired  result.

It is also possible to give a
direct proof, which shows that the problem admits a kernel
of size at most $O(k^2)$. To do so, we replace each linear equation
of at most $r$ variables
by a set of $2^{r-1}$ clauses, so that if the variables $z_i$
satisfy the equation, the corresponding Boolean variables $x_i=(-1)^{z_i}$  satisfy all
these clauses, and if the variables $z_i$ do not satisfy the
equation, then the variables $x_i$ above   satisfy only
$2^{r-1}-1$ of the clauses. This is done as follows.

Consider, first, a linear equation with exactly $r$ variables.
After
renumbering the variables, if needed,  a typical equation
is of the form  $z_1+z_2 + \cdots +z_r=b$, where the sum is over
$\mathbb{F}_2$ and $b \in \{0,1\}$. There are exactly $2^{r-1}$
Boolean assignments $\delta=(\delta_1,\delta_2, \ldots , \delta_r)$
for the variables $z_i$
that do {\bf not} satisfy  the equation. For each such assignment
$\delta$ let $C_{\delta}$ be the clause consisting of $r$ literals,
where the literal number $i$ is $x_i$ if $\delta_i=1$ and is
$\overline{x_i}$ if $\delta_i=0$. Note that if the variables
$z_1,z_2, \ldots ,z_r$ satisfy the  above equation, then
$(z_1,z_2,\ldots ,z_r)$ is not one of the vectors $\delta$
considered, and hence each of the clauses $C_{\delta}$
constructed contains at
least one satisfied literal when $x_i=(-1)^{z_i}$. Therefore, in this
case all clauses are
satisfied. A similar argument shows that if the variables $z_i$ do
not satisfy the equation, there will  be exactly one non-satisfied
clause, namely the one corresponding to the vector
$\delta=(z_1,z_2, \ldots ,z_r)$. 

The construction can be extended
to equations with less than $r$ variables. Indeed,
the only property
used in the transformation above is that there are exactly
$2^{r-1}$ Boolean assignments  for the variables  $z_1,z_2, \ldots
,z_r$ that do not satisfy the equation. If the equation
has only $(1 \leq)~ s<r$ variables, add to these variables an arbitrary set
of $r-s$ of the other variables, and consider the set of all
Boolean assignments to this augmented set of variables that do not
satisfy the equation. Here, too, there are exactly $2^{r-1}$ such
assignments  and we can thus repeat the construction above in this
case as well.

The above procedure transforms a set of $W$ linear equations over
$\mathbb{F}_2$
into a multiset of $2^{r-1} W$ clauses. Moreover, if some truth
assignment  does not satisfy exactly $\ell$ equations, then the same
assignment does not satisfy the same number, $\ell$, of clauses.
In particular, there is an assignment satisfying all equations
but $(W-k)/2$ of them if and only if  there is an assignment
satisfying all clauses but $(W-k)/2$ of them. This means that
among
the $m=2^{r-1}W$ clauses, the number of satisfied ones is
$m-(W-k)/2=[(2^r-1)m+2^{r-1}k]/2^r$. This reduces an instance of
{\MLT} with $W$ equations and parameter $k$ to an instance of {\MST}
with $2^{r-1}W$ clauses and parameter $2^{r-1}k$. Since $r$ is a
constant, this provides the required kernel of size $O(k^2)$,
completing the proof.
\end{proof}

Our algorithm for the problem $\MST$ can be easily modified to provide, efficiently, for any given
instance of $m$ clauses to which there is a truth assignment satisfying
at least $k/2^r$ clauses above the average, an assignment for the
variables with this property.  Indeed, the proof of Theorem~\ref{t95}
only requires that the variables $x_i$ are $4r$-wise independent, and
there are known constructions of polynomial size sample spaces
supporting such random variables (see, e.g., \cite{AS}, Chapter 16).
Thus, if in the polynomial $X$, $\sqrt{|{\cal S}|}/(2\cdot 3^r)\geq k$,
then one can find an assignment satisfying at least as many clauses as
needed by going over all points in such a sample space, and if $\sqrt{|{\cal S}|}/(2\cdot 3^r)< k$, one
can solve the problem by an exhaustive search.


\section{MAX-$2$-SAT}\label{sec4}

In this section we describe an alternative, more combinatorial, approach
to the problem $\MST$ for $r=2$. Although this approach is somewhat more
complicated than the one discussed in the previous section, it provides
an additional insight to this special case of the problem, and allows us
to obtain a kernel with a linear number of variables for \textsc{Max-$2$-Sat${}_{\text{tlb}}$}.

We start with a simple reduction rule that applies to any value of~$r$.

\subsection{The Semicomplete Data Reduction}

Several data reduction rules, or preprocessing methods in a broader sense, have been suggested and applied to {\sc Max-Sat} in the literature (cf. \cite{SZ05} and references therein). We are especially interested in reduction rules which simplify 2-CNF formulas. Most preprocessing methods for 2-CNF formulas are not applicable to our type of parameterization. For example, a data reduction rule known as {\em Resolution Rule}, replaces two clauses $xy$, $\ol{x}z$ by $yz$ if the literals $x$ or $\ol{x}$ does not appear in any other clauses. Resolution Rule is sound under the parameterization of {\sc Max-2-Sat}, where we ask if all except for at most $k$ clauses can be satisfied. However, this reduction rule is not applicable to our type of parameterization. On the other hand, there is a data reduction rule introduced in \cite{SZ05}, where it was called {\em Complementary Unit Rule}, which is applicable to our parameterization of  {\sc Max-2-Sat}. Below we use this reduction rule.

We say that a pair of distinct clauses $Y$ and $Z$ has a \emph{conflict}
if there is a literal $p\in Y$ such that $\ol{p}\in Z$. We say that an
$r$-CNF formula $F$ is \emph{semicomplete} if the number of clauses is
$m=2^r$ and every pair of distinct clauses of $F$ has a conflict.  A
semicomplete $r$-CNF formula is \emph{complete} if each clause is over
the same set of variables.  There are $r$-CNF formulas that are
semicomplete but not complete; consider for example $\{x y$, $x\ol{y}$,
$\ol{x}z$, $\ol{x} \ol{z}\}$. We have the following:

\begin{lemma}\label{lem:semic}
  Every truth assignment to a semicomplete $r$-CNF formula satisfies
  exactly $2^r-1$ clauses.
\end{lemma}

\begin{proof}
  Let $S$ be a semicomplete $r$-CNF formula.  To prove that no truth
  assignment satisfies all clauses of $S$ we use the following simple
  counting argument from~\cite{Iwama89}.  Observe that every clause is
  not satisfied by exactly $2^{n-r}$ truth assignments. However, each of
  these assignments satisfies each other clause (due to the
  conflicts). So, we have exactly $2^r \cdot 2^{n-r}$ truth assignments
  not satisfying $S$. But $2^r \cdot 2^{n-r}=2^n$, the total number of
  truth assignments.

  Now let $\tau$ be a truth assignment of $S$. By the above, $\tau$ does
  not satisfy a clause $C$ of $S$. However, $\tau$ satisfies any other
  clause of $S$ as any other clause has a conflict with~$C$.
\end{proof}

\noindent Consider the following data reduction procedure.

Given an $r$-CNF formula $F$ that contains a semicomplete subset
$F'\subseteq F$, delete $F'$ from $F$ and consider $F\setminus F'$
instead.  Let $F^S$ denote the formula obtained from $F$ by applying
this deletion process as long as possible.  We say that $F^S$ is
obtained from $F$ by \emph{semicomplete reduction}.

We state the following two simple observations as a lemma.

\begin{lemma}\label{lem:redsafe}
  Let $F$ be an $r$-CNF formula.
  \begin{enumerate}
  \item  $F^S$ can be obtained from $F$ in polynomial time.
  \item $\sat(F)-\sat(F^S)=(1-2^{-r})(\Card{F}-\Card{F^S})$.
  \end{enumerate}
\end{lemma}

\subsection{Kernelization}

Let $F$ be a 2-CNF formula.  A variable $x\in \var(F)$ is
\emph{insignificant} if for each literal $y$ the numbers of occurrences
of the two clauses $xy$ and $\ol{x}y$ in $F$ are the same. A variable
$x\in \var(F)$ is \emph{significant} if it is not insignificant. A
literal is significant or insignificant if its underlying variable is
significant or insignificant, respectively.
\begin{theorem}\label{the:sigkernel}
  Let $F$ be a 2-CNF formula with $F=F^S$ (i.e., $F$ contains no
  semicomplete subsets) and let $k\geq 0$ be an integer. If $F$ has more
  than $3k-2$ significant variables, then $\sat(F)\geq (3\Card{F}+k)/4$.
\end{theorem}

The remainder of this section is devoted to the proof of
Theorem~\ref{the:sigkernel} and its corollary. Let $F$ be a 2-CNF formula with $m$ clauses and $n$ variables and let $k$ be an integer. We assume that $F$ contains no semicomplete subsets, i.e., $F=F^S$.

For a literal $x$ let $c(x)$ denote the number of clauses in $F$
containing $x$. Given a pair of literals $x$ and $y$, $x\neq \ol{y}$,
let $c(xy)$ be the number of occurrences of clause $xy$ in~$F$.

Given a clause $C\in F$ and a variable $x\in \var(F)$, let $\delta_C(x)$
be an indicator variable whose value is set as $\delta_C(x)=1$ if $x\in
C$, $\delta_C(x)=-1$ if $\ol{x}\in C$, and $\delta_C(x)=0$ otherwise.

\begin{lemma}\label{lem:above}
  For each subset $R=\{x_1,\ldots ,x_q\}\subseteq \var(F)$ we have
  $\sat(F)\geq (3m +k_R)/4$ for
  \[
  k_R= \sum_{1\leq i \leq q} (c(x_i)-c(\ol{x_i})) +
  \sum_{1\leq i < j \leq q}
  \big(
 c(x_i\ol{x_j})+c(\ol{x_i}x_j)
  -c(x_ix_j)-c(\ol{x_i}\ol{x_j})
  \big).
\]
\end{lemma}
\begin{proof}
  Take a random truth assignment $\tau\in 2^{\var(F)}$ such that
  $\tau(x_i)=1$ for all $i\in \{1,\dots,q\}$ and
  $\Prob(\tau(x)=1)=0.5$ for all $x\in \var(F)\setminus R$. A simple
  case analysis yields that the probability that a clause $C\in F$ is
  satisfied by $\tau$ is given by
  \[
  \Prob(\text{$\tau$ satisfies $C$})=1-\frac{1}{4}\prod_{1\leq i
    \leq q}(1-\delta_C(x_i)).
  \]
  \sloppypar Observe that for any clause $C$ and any three distinct
  variables $x,y,z$ we have $\delta_C(x) \delta_C(y)\delta_C(z)=0$ as
  $\var(C)$ contains exactly two variables.  Hence we can determine the
  expected number of clauses satisfied by $\tau$ as follows.
\begin{eqnarray*}
  \Exp(\sat(\tau,F))
  &=&
  \sum_{C\in F} \Prob[\;\text{$\tau$ satisfies $C$}\;]\\
  &=&
  \sum_{C\in F}\big\{1-\frac{1}{4}\prod_{1\leq i \leq q}(1-\delta_C(x_i))\big\} \\
  &=&
  \frac{3}{4}m +
  \frac{1}{4}
  \sum_{C\in F}\big\{
  \sum_{1\leq i \leq q} \delta_C(x_i)
  -\sum_{1\leq i < j \leq q}\delta_C(x_i)\delta_C(x_j)\big\} \\
  &=&
  \frac{3}{4}m +
  \frac{1}{4} \big\{
  \sum_{1\leq i \leq q} \sum_{C \in F} \delta_C(x_i)
  - \sum_{1\leq i < j \leq q} \sum_{C \in F}\delta_C(x_i)\delta_C(x_j)\big\}  \\
  & =&
  \frac{3}{4}m +
  \frac{1}{4}k_R. \hspace{7cm}\qedhere
\end{eqnarray*}
\end{proof}

\medskip\noindent It is noteworthy that $\Prob(\text{$\tau$ satisfies $C$})=1-\frac{1}{4}\prod_{1\leq i \leq q}(1-\delta_C(x_i))$ in the proof of Lemma \ref{lem:above} is similar to a term of $X$ defined in Section \ref{algerep}. The term $1-\prod_{x_i\in \var(C)}(1+\epsilon_ix_i)$ of $X$ returns a fixed value on $C$ for a given ({\em fully determined}) truth assignment, depending on whether $C$ is satisfied or not. Similarly, the term $1-\frac{1}{4}\prod_{1\leq i \leq q}(1-\delta_C(x_i))$ returns a probability of $C$ being satisfied for a given ({\em partially determined}) random truth assignment. The benefit of having a probabilistic form of $X$ is that we now have a way to ignore a large number of variables, e.g., $V\setminus R$ in Lemma \ref{lem:above}, instead of searching for a fully determined truth assignment so as to compute $X$. For the case $r=2$, this probabilistic form of $X$ can be immediately interpreted in a graph-theoretical language as will be shown below.

Due to Lemma \ref{lem:above}, the task is now reduced into finding a subset $R$ of variables with $k_R\geq k$. These are variables which form the deterministic part of a partially random truth assignment. Using a notion of {\em switch} defined later, we replace $F$ by an equivalent formula in which every variable of $R$ is set to 1.
To find $R$ we use a graph-theoretical approach introducing an auxiliary weighted graph in which we seek an induced subgraph of weight at least $k$.
In particular, we note that an `independent' structure of an induced subgraph ensures its weight to be above a certain bound growing with the size of the induced subgraph. This means that if $(F,k)$ is a {\sc No}-instance, we do not have a large `independent' structure. Using  this fact and the Tutte-Berge formula for the size of a maximum matching in a graph (stated after Lemma \ref{collectionstar}), we will prove an upper bound on the number of vertices in the auxiliary weighted graph.

\medskip\noindent We construct an \emph{auxiliary graph} $G=(V,E)$ from
$F$ by letting $V=\var(F)$ and $xy\in E$ if and only if there exists a
clause $C\in F$ with $\var(C)=\{x,y\}$ (equivalently, $c(x \ol{y}) +
c(\ol{x}y) + c(xy) + c(\ol{x} \ol{y})\geq 1$).

We assign a \emph{weight} to each vertex $x$ and edge $xy$ of
$G=(V,E)$:
\begin{eqnarray*}
  w(x)&:=&\sum_{C \in F} \delta_C(x)=c(x)-c(\ol{x}), \\
  w(xy)&:=& -\sum_{C \in F}\delta_C(x)\delta_C(y)
  =c(x \ol{y}) + c(\ol{x}y)
  - c(xy) - c(\ol{x} \ol{y}).
\end{eqnarray*}
For subsets $U\subseteq V$ and $H\subseteq E$, let $w(U)=\sum_{x\in
  U}w(x)$ and $w(H)=\sum_{xy\in H}w(xy)$. The \emph{weight} $w(Q)$ of a
subgraph $Q=(U,H)$ is $w(U)+w(H)$. Let $G^0$ be the graph obtained from
$G$ by removing all edges of weight zero.
\begin{lemma}\label{isol}
  A variable $x\in \var(F)$ is insignificant if and only if $x$ is an
  isolated vertex in $G^0$ and $w(x)=0$.
\end{lemma}
\begin{proof}
  Suppose $x\in \var(F)$ is insignificant.  Choose an edge $xy\in E$
  (this is possible since by construction $G$ has no isolated vertices).
  Since $x$ is insignificant, $c(x \ol{y})=c(\ol{x} \ol{y})$ and $c(xy)
  = c(\ol{x}y)$ and thus $w(xy)=0$. Therefore the edge $xy$ does not
  appear in $G^0$ and $x$ is isolated in $G^0$. Observe that we have
  $c(x)=c(\ol{x})$, which implies $w(x)=0$.

  Suppose $x\in \var(F)$ is an isolated vertex of $G^0$ and
  $w(x)=0$. Since $G$ has no isolated vertices, we have $w(xy)=0$ for
  all $xy\in E$. In order to derive a contradiction, let us suppose $x$
  is a significant variable of $F$. Consequently there is (i)~either a
  clause $xy\in F$ such that $c(xy)>c(\ol{x}y)$, or (ii)~there is a
  clause $\ol{x}y\in F$ such that $c(\ol{x}y)>c(xy)$. We consider
  case~(i) only, case~(ii) can be treated analogously.  With $w(xy)= 0$,
  we have $c(x\ol{y})>c(\ol{x}\ol{y})$, and thus $x\ol{y}\in F$.

  Now the condition $w(x)=c(x)-c(\ol{x})=0$ implies the existence of an
  edge $xz\in E$ with $z\neq y$ such that for some $z'\in \{z,\ol{z}\}$
  we have $\ol{x}z'\in F$ and $c(\ol{x}z')>c(xz')$. Without loss of
  generality, assume that $z'=z$. Since $w(xz)=0$, we have
  $\ol{x}\ol{z}\in F$. However, the four clauses $xy$, $x\ol{y}$,
  $\ol{x}z$, $\ol{x}\ol{z}$~in $F$ form a semicomplete 2-CNF formula,
  which contradicts our assumption that $F=F^S$.  Hence $x$ is indeed an
  insignificant variable.
\end{proof}

For a set $X \subseteq \var(F)$ we let $F_X$ denote the 2-CNF formula
obtained from $F$ by replacing $x$ with $\ol{x}$ and $\ol{x}$ with $x$
for each $x\in X$. We say that $F_X$ is obtained from $F$ by
\emph{switching}~$X$.

The following lemma follows immediately from the definitions of switch
and weights.

\begin{lemma}\label{lem:switch}
  The auxiliary graph $G_X$ corresponding to $F_X$ can be obtained from
  $G=(V,E)$ by reversing the signs of the weights of all vertices in $X$
  and all edges between $X$ and $V\setminus X$.  Moreover,
  $\sat(F)=\sat(F_X)$.
\end{lemma}

To distinguish between weights in $G$ and $G_X$, we use $w_X(.)$ for
weights of $G_X$. Similarly, we use $c_X(.)$ for~$F_X$.

It is sometimes convenient to stress that the set $X$ we are switching
induces a subgraph. We can \emph{switch an induced graph} $Q$ by
switching all the vertices of~$Q$. Observe that by switching an induced
graph $Q$, we reverse the signs of weights on all vertices of $Q$ and
all edges incident with exactly one vertex of $Q$, but the sign of each
edge within $Q$ remains unchanged. This property will play a major role
to show that a certain structure meets the condition of the following
lemma.

\begin{lemma}\label{lem:subgraph}
  If there exist a set $X\subset V(G^0)$ and an induced subgraph
  $Q=(U,H)$ of $G^0$ with $w_X(Q)\geq k$, then $\sat(F)\geq (3m+k)/4$.
\end{lemma}
\begin{proof}
  We consider $U=\{x_1,\dots,x_q\}$ as a subset of $\var(F_X)$. By
  Lemmas~\ref{lem:above} and~\ref{lem:switch}, $\sat(F)=\sat(F_X)\geq
  (3m +k_U)/4$, where
  \begin{eqnarray*}
 k_U \!\!&=&\!\!  \sum_{i=1}^q (c_X(x_i)-c_X(\ol{x_i})) +
     \!\! \sum_{1\leq i < j \leq q}(
     c_X(x_i \ol{x_j}) + c_X(\ol{x_i}x_j)
     - c_X(x_ix_j) - c_X(\ol{x_i} \ol{x_j}) )  \\
   &=&\!\! \sum_{i=1}^q w_X(x_i) + \sum_{1\leq i < j \leq q}w_X(x_ix_j)
   \quad = \quad w_X(Q)  \quad\geq\quad k.
\hspace{0.4cm}\qedhere
  \end{eqnarray*}
\end{proof}

To apply Lemma~\ref{lem:subgraph} in the proof of
Theorem~\ref{the:sigkernel}, we will focus on a special case of induced
subgraphs of $G^0$. For a set $U\subseteq V(G^0)$, let $G^0[U]$ denote
the subgraph of $G^0$ induced by $U$. We call $G^0[U]$ an {\em induced
  star} with {\em center} $x$ if $x$ is a vertex of $G^0$, $I$ is an
independent set in the subgraph of $G^0$ induced by the neighbors of $x$
and $U=\{x\}\cup I$. We are interested in the induced star due to the
following property.

\begin{lemma}\label{inducedstar}
  Let $x$ be the center of an induced star $Q=G^0[U]$ and let
  $I=U\setminus \{x\}$.  Then there is a set $X\subseteq U$ such that
  $w_X(Q)\ge |I|$.
\end{lemma}
\begin{proof}
  Let $H$ be the set of edges of $Q$. We may assume that $w(xy)\geq 0$ for
  each $y \in I$ since otherwise we can switch $y$, and $w(xy$) is integral. By a \emph{random
    switch} of $Q$, we mean a switch of every vertex of $Q$ with probability $0.5$. Take
  a random switch $R$ of~$Q$. Then we have $\Exp(w_R(z))=0$ for all
  $z\in U$. Note that the sign of each edge in $H$ remains
  positive. Hence we have $\Exp(w_R(Q))= w(H) \geq |I|$ and thus there
  exists a set $X\subseteq U$ for which $w_X(Q)\geq |I|$.
\end{proof}

If we are given more than one induced star, a sequence of random
switches gives us a similar result.

\begin{lemma}\label{collectionstar}
  Let $Q_1=(U_1,H_1),\ldots ,Q_m=(U_m,H_m)$ be a collection of
  vertex-disjoint induced stars of $G^0$ with centers $x_1,\ldots ,x_m$,
  let $U=\bigcup_{i=1}^{m}U_i$, and let $Q=G^0[U]$.  Then there is a set
  $X\subseteq U$ such that $w_X(Q)\ge \sum_{i=1}^m |I_i|,$ where
  $I_i=U_i\setminus \{x_i\},$ $i=1,\ldots ,m$.
\end{lemma}
\begin{proof}
  As in the proof of Lemma~\ref{inducedstar}, we may assume that all the
  edges of $H_i$ have positive weights.  Let $H$ be the set of edges of
  $Q$.  By a \emph{random switch} of $Q$, we mean a sequence of switches
  of $Q_1,\ldots ,Q_m$ each with probability $0.5$. Take a random switch
  $R$ of~$Q$. Then we have $\Exp(w_R(x))=0$ for all $x\in U$. Moreover,
  for the subgraph $Q$ of $G^0_R$, it holds that $\Exp(w_R(xy))=0$ for
  all $xy\in H\setminus \bigcup_{i=1}^{m}H_i$ since each choice of
  $w_R(xy)\geq 0$ and $w_R(xy)\leq 0$ is equally likely. By linearity of
  expectation and Lemma~\ref{inducedstar}, we have $\Exp(w_R(Q))=
  w(\bigcup_{i=1}^{m}H_i) \geq \sum_{i=1}^m |I_i|$ and thus there exists
  a set $X\subseteq U$ for which $w_X(Q)\geq \sum_{i=1}^m |I_i|$.
\end{proof}

Note that we can derandomize the procedures suggested in the proofs of Lemma \ref{inducedstar} and \ref{collectionstar} using the standard technique of conditional expectation \cite{AS}.

We are now in the position to complete the proof of
Theorem~\ref{the:sigkernel}.

Suppose that $(F,k)$ is a no-instance, i.e., $\sat(F)< (3m+k)/4$.
Notice that a matching can be viewed as a collection of induced stars of
$G^0$ for which $|I_i|=1$. It follows by Lemmas~\ref{lem:subgraph} and
\ref{collectionstar} that $G^0$ has no matching of size $k$.  The
Tutte-Berge formula \cite{Berge58,BondyMurty08} states that the size of
a maximum matching in $G^0$ equals
\[
\min_{S\subseteq V(G^0)}\frac{1}{2}\{ |V(G^0)| + |S| - oc(G^0-S)\}
\]
where $oc(G^0-S)$ is the number of odd components (connected components
with an odd number of vertices) in $G^0-S$. Hence there is a set
$S\subseteq V(G^0)$ such that $|V(G^0)| + |S| - oc(G^0-S) < 2k$. It
follows that
\begin{equation}\label{eq1}
  |V(G^0)|\le oc(G^0-S)-|S|+2k-1.
\end{equation}

We will now classify odd components in $G^0-S$. One obvious type of odd
components is an isolated vertex in $G^0$ of weight zero, which
corresponds to an insignificant variable by Lemma~\ref{isol}. All the
other odd components can be categorized into one of the following two
types:

\begin{enumerate}
\item Let $Q_1,\ldots, Q_L$ be the odd components of $G^0-S$ such that
  for all $1\leq i \leq L$ we have $|Q_i|=1$ and $Q_i$ is a significant
  variable.

\item Let $Q'_1,\ldots, Q'_{L'}$ be the odd components of $G^0-S$ such
  that for all $1\leq i \leq L'$ we have $|Q'_i|>1$.
\end{enumerate}

We construct a collection of induced stars as follows. From each of
$Q'_1,\ldots, Q'_{L'}$ we choose an edge, which is an induced star with
$|I|=1$. Let us consider $Q_1,\ldots, Q_L$. Each vertex $Q_i$ is
adjacent to at least one vertex of $S$. Thus, we can partition
$Q_1,\ldots, Q_L$ into $|S|$ sets, some of them possibly empty, such
that each partite set forms an independent set in which every vertex is
adjacent to the corresponding vertex $x_i$ of $S$. Each partite set,
together with $x_i$, forms an induced star. Now observe that we have a
collection of induced stars and the total number of edges equals
$L+L'$. If $L+L' \ge k$, Lemma~\ref{collectionstar} implies that for
some set $X$ of vertices from the odd components $w_X(Q)\geq k$, which
is impossible by Lemma~\ref{lem:subgraph}. Hence $L+L'\leq k-1$.

Therefore, $oc(G^0-S)-n' = L + L' \leq k-1$,
where $n'$ is the number of insignificant variables. By (\ref{eq1}),
we have $|V(G^0)|-n'\le k-1-|S|+2k-1\le 3k-2$. It remains to observe that $|V(G^0)|-n'$ equals the number of
significant variables of~$F$. This completes the proof of
Theorem~\ref{the:sigkernel}.

\begin{corollary}\label{the:kern}
  The problem {\MTT} admits a (polynomial time) reduction to a problem
  kernel with at most $3k-1$ variables.
\end{corollary}
\begin{proof}
  Consider an instance $(F,k)$ of the problem.  First we apply the
  semicomplete reduction and obtain (in polynomial time) an instance
  $(F',k)$ with $F'=F^S$. We determine (again in polynomial time) the
  set $S'$ of significant variables of $F'$. If $\Card{S'}> 3k-2$ then
  $(F',k)$ is a yes-instance by Theorem~\ref{the:sigkernel}, and
  consequently $(F,k)$ is a yes-instance by Lemma~\ref{lem:redsafe}.
  Assume now that $\Card{S'}\leq 3k-2$.

  Let $z$ be a new variable not occurring in $F$. Since $F'=F^S$, no
  clause contains two insignificant variables and, thus, each
  insignificant variable can be replaced by $z$ without changing the
  solution to $(F',k)$. Let us denote the modified $F'$ by $F''$; $F''$
  has at most $3k-1$ variables.

  Let $p$ be the number of clauses in $F''$. Observe that we can find a
  truth assignment satisfying the maximum number of clauses of $F''$ in
  time $O(p8^k)$. Thus, if $p>8^k$, we can find the optimal truth
  assignment in the polynomial time $O(p^2)=O(m^2)$. Thus, we may assume
  that $F''$ has at most $8^k$ clauses. Therefore, $F''$ is a kernel of
  the {\MTT} problem.
\end{proof}

\section{Extension to Boolean Constraint Satisfaction Problems}\label{sec5}

The fixed-parameter tractability result on {\MST}
can be easily extended to any family of Boolean $r$-Constraint
Satisfaction Problems.  Here is an outline of the argument.

Let $r$ be a fixed positive integer, let
$\Phi$ be a set of Boolean functions, each involving at
most $r$ variables, and let
${\cal F}=\{f_1, f_2,
\ldots, f_m\}$ be a collection of Boolean functions, each being a
member of $\Phi$, and each acting on some subset of
the $n$ Boolean variables $x_1,x_2, \ldots ,x_n$.
The Boolean Max-$r$-Constraint
Satisfaction Problem (corresponding to $\Phi$), which we denote by
the {\sc Max-$r$-CSP} problem, for short, when $\Phi$ is clear from
the context, is the
problem of finding a truth assignment to the variables so as to maximize
the total number of functions satisfied. Note that this includes, as a
special case, the \textsc{Max-$r$-Sat} problem considered in the
previous section, as well as many related problems.  As most interesting
problems of this type are NP-hard, we consider their parameterized
version, where the parameter is, as before, the number of functions
satisfied minus the expected value of this number.  Note, in
passing, that the above expected value is a tight lower
bound for the problem, whenever the family ${\Phi}$ is closed under
replacing each variable by its complement, since if we apply any
Boolean function to all $2^r$ choices of literals whose underlying
variables are any fixed set of $r$ variables, then any truth
assignment  to the variables satisfies exactly the same number of
these $2^r$ functions.

For each Boolean
function $f$ of $r(f)$ Boolean variables
\[
x_{i_1}, x_{i_2}, \ldots , x_{i_{r(f)}},
\]
define a random variable $X_f$ as follows. As in the discussion of the
\textsc{Max-$r$-Sat} problem, suppose each variable $x_{i_j}$ attains
values in $\{-1,1\}$.  Let $V \subseteq \{-1,1\}^{r(f)}$ denote the set of all
satisfying assignments of $f$. Then
\[
X_f(x_1, x_2, \ldots ,x_n) = \sum_{v=(v_1, \ldots ,v_{r(f)}) \in V}
2^{r-r(f)}[\prod_{j=1}^{r(f)} (1+x_{i_j} v_j) - 1].
\]

This is a random variable defined over the space $\{-1,1\}^n$ and its
value at $x=(x_1,x_2, \ldots ,x_n)$ is $2^r-|V|\cdot 2^{r-r(f)}$ if $x$ satisfies $f$,
and is $-|V|\cdot 2^{r-r(f)}$ otherwise.  Thus, the expectation of $X_f$ is zero. Define
now $X =\sum_{f \in {\cal F}} X_f$.  Then the value of $X$ at $x=(x_1,
x_2, \ldots ,x_n)$ is precisely $2^r(s-a)$, where $s$ is the number of
the functions satisfied by the truth assignment $x$, and $a$ is the
average value of the number of satisfied functions.  Our objective is to
decide if $X$ attains a value of at least $k$. As this is a polynomial
of degree at most $r$ with integer coefficients and expectation zero, we
can repeat the arguments of Section~\ref{sec3} and prove that, for every
fixed $r$, the problem is fixed-parameter tractable.
Moreover, our previous arguments show that the problem admits
a polynomial-size bikernel reducing it to  an instance of {\MLT} of
size $O(k^2)$, and if the specific $r$-CSP problem considered is
NP-complete, then there is a polynomial
size kernel. This is the case for most interesting choices of
the family ${\Phi}$.

\section{Discussions}\label{sec:d}

Our results are mainly of theoretical interest. However, our kenelization for $\MTT$ might be of some practical interest for families of istances of \textsc{Max-$2$-Sat} where the maximum number of satisfied clauses is close to $3m/4.$

Recently Crowston et al. \cite{CGJKR} proved that $\MTT$ has a kernel with $O(k \log k)$ variables for every fixed $r\ge 2$. The new result uses several ideas given in this paper, but employes linear algebraic rather than probabilistic tools.

\medskip
\paragraph{Acknowledgments}
Research of Alon was partially supported by an ERC Advanced grant.
Research of Gutin, Kim and Yeo was supported in part by an EPSRC
grant. Research of Gutin was also supported
in part by the IST Programme of the European Community, under the
PASCAL 2 Network of Excellence.

\urlstyle{rm}

\end{document}